\documentclass{article}
\usepackage[utf8]{inputenc}
\usepackage{amsmath}
\newcommand{\dv}[2]{\frac{\mathrm{d} #1 }{\mathrm{d} #2}}

\newcommand{\cmt}[1]{\quad \text{#1}}
\usepackage{tikz}
\usepackage{pgfplots}
\usepackage{amsfonts}
\usepackage{amsthm}
\definecolor{linkcolor}{rgb}{0,0,0.6} 
\usepackage[ pdftex,colorlinks=true,
pdfstartview=FitV,
linkcolor= linkcolor,
citecolor= linkcolor,
urlcolor= linkcolor,
hyperindex=true,
hyperfigures=false]
{hyperref}
\usepackage{float}

\newtheorem{theorem}{Theorem}[section]
\newtheorem{corollary}{Corollary}[theorem]
\newtheorem{lemma}[theorem]{Lemma}
\newtheorem{axiom}{Axiom}[section]

\theoremstyle{definition}
\newtheorem{definition}{Definition}[section]

\theoremstyle{remark}
\newtheorem*{remark}{Remark}

\title{About constant-product automated market makers}
\author{
  Théodore Conrad\\
  \texttt{École Normale Supérieure}
  \and
  Guillaume Méroué\\
  \texttt{École Normale Supérieure de Lyon}
  \and
  Arthur Vinciguerra\\
  \texttt{École Normale Supérieure de Lyon}
}

\date{September 2022}

\begin{document}

\maketitle


\begin{abstract}
    Constant-product market making functions were first introduced by Hayden Adams in 2017 to create Uniswap, a decentralised exchange on Ethereum. This enables users to exchange assets at any given rate. Some variations such as Balancer and Curve were later introduced. In this paper, we analyse the maths that rule this type of protocol. We show that splitting a trade in multiple smaller trades does not impact the final exchange rate. We also show that the protocol is safer and more profitable when no one recompounds their fees.
\end{abstract}

\section{Introduction}

In October 2008, Satoshi Nakamoto published the Bitcoin whitepaper \cite{bitcoin}: the first \textit{cryptocurrency}, a decentralised currency that no one can control. Since Satoshi Nakomoto's revolutionnary paper, a lot of effort has been made to try decentralising things. In 2014, Vitalik Buterin published the whitepaper \cite{ethereum} of his own blockchain: Ethereum the first one implementing \textit{smart contracts}, a piece of code that is run on the blockchain. This technological breakthrough enabled the making of the first Decentralised Finance (DeFi) protocol MakerDao \cite{makerdao}. The latter enabled users to make collateralized debt position using cryptocurrencies, ie. deposit cryptocurrency to borrow money. Since then, the number of DeFi protocols skyrocketted and the total value locked in DeFi smart contracts reached a $\$250$b all time high \cite{defilama} in december 2021. In this paper, we focus on one specific type of DeFi protocol: decentralised exchanges. The first decentralised exchange protocol, Uniswap \cite{UniswapV1}, was proposed by Hayden Adams in 2017 and was released in 2018 on the Ethereum blockchain. Uniswap is an automated market maker (AMM): users can lock cryptocurrencies on a smart contract, which will be used as liquidity to market make. Uniswap uses a \textit{constant product} market making model, which we discuss later. This model works well in general but is not optimal for token pairs which exchange rate is stable. For this purpose, Michael Egorov created another protocol, Curve \cite{Curve}, that uses another maket making model that focuses on stable token pairs.\\
The betanet version of Beaker uses a Uniswap v2 \cite{UniswapV2} model with two changes: fees are not automatically compounded in the pool and users can trade with a given spread. In this paper, we break down the maths behind this protocol and explore various questions that can arise when interacting with it.

\newpage
\tableofcontents
\newpage

\newpage
\section{Maths for a constant-product AMM}
In this section, we explore the maths behind the fixed product model. In this model, pools containing two assets are set up and any individual (called a liquidty provider) can provide liquidity to it. These pools act as markets makers and provide liquidity at any given rate. More elaborate models enable each individual to provide liquidity for a given price range \cite{UniswapV3}. \\ 

In the rest of the paper we consider a pool made of tokens $\mathbf{X}$ and $\mathbf{Y}$.
Let $x$ (resp. $y$) be the amount of tokens $\mathbf{X}$ (resp. $\mathbf{Y}$) in the pool, and let $r$ be the exchange rate for the pool.

We start by choosing the definition of \textit{liquidity} and the \textit{exchange rate} of our pool. These define how our pool is going to market make.

There are two equations that are at the center of the model. The first one states that the product of the number of tokens should be a constant and is often called \textit{fixed product} equation. In mathematical terms, this means that the geometric mean should always stay constant. This geometric mean is what we call \textit{liquidity}.

\begin{definition}
The liquidity $L$ of the pool is defined by the equation:
\[
    L = \sqrt{xy} .
\]
\end{definition}

Therefore, the fixed product assumption just means that the liquidity is constant during a trade.
\begin{axiom}
$L$ is an invariant of a trade.
\end{axiom}

Then, we must decide how our pool is going to make the market ie. how it is going to choose the exchange rate. A natural idea is to choose $r = -\frac{\partial x}{\partial y}$. In fact, $r$ is the number of tokens $\mathbf{X}$ we receive when we add an infinitesimal amount of tokens $\mathbf{Y}$ to the pool. For simplicity, we use the result of the previous computation as the definition of $r$.
\begin{definition}
We set the exchange rate of the pool to be :
\[
    r = \frac{x}{y}.
\]
\end{definition}

Now that we have defined these concepts, we can define what is the value of the pool. 
\begin{definition}[Value of the pool]
The value $V$ of the pool is defined by the equation:
\[
    V = p_xx + p_yy,
\]
where $p_x$ and $p_y$ are the market prices (in $\$$ by token) of the tokens $\mathbf{X}$ and $\mathbf{Y}$.
\end{definition}

To simplify a bit, we make a final assumption on the exchange rate of the pool.

\begin{axiom}[Arbitrage assumption]
We assume that the exchange rate of the pool is always equal to the exchange rate of the market: 
\[
    \frac{p_y}{p_x} = \frac{x}{y} .
\]
\end{axiom}
In reality, when the exchange rate of the pool is not equal to the one of the market, arbitrageurs can make a small profit by restoring the equality. Therefore, this assumption holds in practice.

In the next subsections, we will give various relations between the variables of the pool and we will describe how these variables change when a trade is done. We later explore the notion of \textit{Impermanent Loss}, which characterizes the risks of a liquidity provider and we conclude by evaluating the evolution of the portfolio of liquidity providers. 

\pagebreak

\subsection{Some useful relations}
In this subsection, we will explore various relations involving $x$, $y$, $L$, $r$ and $V$.

\begin{lemma}[Dependency on $r$]
For all $r\in ]0,+\infty[$ and $L \in [0,+\infty[$ : 
\begin{align*}
    x &= L\sqrt{r}\\
    y &= \frac{L}{\sqrt{r}} .
\end{align*}
\end{lemma}

\begin{proof}
This is straightforward from the definitions of $L$ and $r$. 
\end{proof}

\begin{lemma}[Expression involving $V$]
We have the following relationships : 
\begin{align*}
    x &= \frac{1}{2}\frac{V}{p_x} \\
    y &= \frac{1}{2}\frac{V}{p_y} .
\end{align*}
\end{lemma}
\begin{proof}
Using the definition of $V$ and the arbitrage assumption, we get :
\begin{align*}
    V &= p_xx + p_yy \\
      &= p_xx + p_xx \\
      &= 2p_xx
\end{align*}
The same can be done for $y$.
\end{proof}

\begin{lemma}[Expression for L]
The liquidity of the pool is linked to its value and to the prices of $\mathbf{X}$ and $\mathbf{Y}$ by the following expression:
\[
    L = \frac{1}{2}\frac{V}{\sqrt{p_yp_x}} .
\]
\end{lemma}
\begin{proof}
This results from $L^2 = xy$.
\end{proof}

These 3 expressions give various ways of determining $x$, $y$, $r$, $L$ and $V$.

\pagebreak

\subsection{Evolution of the pool with a swap}
In this subsection, we try to describe the evolution of $x$, $y$, $r$ and $V$ when a trade is made by the pool. To simplify the notations, we will call $x$ (resp. $y$) the number of tokens $\mathbf{X}$ (resp. $\mathbf{Y}$) before the swap and $x'$ (resp. $y'$) the number of tokens $\mathbf{X}$ (resp. $\mathbf{Y}$) after the swap. We similarly define $r$ and $r'$.

\begin{theorem}[State of the variables after a swap]
After a swap, we have the following equations :
\begin{align*}
    x'y' &= xy \\
    r' &= r \left(\frac{y}{y'}\right)^2
\end{align*}
\end{theorem}
\begin{proof}
The first equation directly follows from $x' y' = L^2 = x y$. \\
The second one is also straightforward to get : 
\begin{align*}
    r' &= \frac{x'}{y'} \\
        &= \frac{1}{y'}\frac{xy}{y'} \\
        &= \frac{(ry) y}{(y')^2} \\
        &= r \left(\frac{y}{y'}\right)^2
\end{align*}
\end{proof}

Before describing the equation that models how a trade takes place, we give a final definition. 

\begin{definition}
The spread $\sigma$ of a transaction is defined as the relative rate change during a transaction 
\[
    \sigma = \left| \frac{r' - r}{r} \right| .
\]
\end{definition}

\begin{corollary}[Swap equation for $\mathbf{Y}$]
Suppose that a user wants to trade $n$ tokens $\mathbf{Y}$ for some tokens $\mathbf{X}$ with a maximum spread $\sigma$, then the user will receive a quantity $m$ of tokens $\mathbf{X}$ given by the following equation : 
\[
    m = \frac{x\alpha}{y + \alpha}, \cmt{with $\alpha = \min \left(n, y\left( \frac{1}{\sqrt{1 - \sigma}} - 1 \right) \right)$}
\]
\end{corollary}
\begin{proof}
By definition of spread, we have a restriction on the maximum amount $q$ of tokens $\mathbf{Y}$ that can be traded. When the spread is limiting, we have : 
\begin{align*}
        \sigma &= \frac{r - r'}{r} \\
        \sigma &= 1 - \left(\frac{y}{y'}\right)^2 . 
\end{align*}
Let $q_\sigma$ be the maximum amount of tokens $y$ that can be traded with spread $\sigma$.
Injecting $y' = y + q_{\sigma}$ in the previous equation gives : 
\[
     1 - \sigma = \left(\frac{y}{y + q_{\sigma}}\right)^2 
\]
which finally gives the following equation for $q_{\sigma}$ : 
\[
    q_\sigma = y\left( \frac{1}{\sqrt{1 - \sigma}} - 1 \right) .
\]
 Therefore, the maximum amount $\alpha$ of tokens $\mathbf{X}$ that the user will swap is 
 \[
     \alpha = \min \left(n, y\left( \frac{1}{\sqrt{1 - \sigma}} - 1 \right) \right) .
 \]
Thus, we have $y' = y + \alpha$ and we can get the number of tokens $m$ that the user will receive: 
\begin{align*}
    m &= x - x' \\ 
        &= x - \frac{xy}{y'} \\
        &= x\left(  1 - \frac{y}{y'}\right) \\
        &= \frac{x\alpha}{y + \alpha} \\
\end{align*}
\end{proof}

\begin{corollary}[Swap equation for $\mathbf{X}$]
Suppose that a user wants to trade $n$ tokens $\mathbf{X}$ for some tokens $\mathbf{Y}$ with a maximum spread $\sigma$, then the user will receive a quantity $m$ of tokens $\mathbf{Y}$ given by the following equation:
\[
    m = \frac{y\alpha}{x + \alpha}, \cmt{with $\alpha = \min \left(n, x\left( \sqrt{1 + \sigma} - 1 \right) \right)$}
\]
\end{corollary}
\begin{proof}
Because $r' > r$, the expression of the spread in this case is : 
\begin{align*}
    \sigma &= \frac{r'}{r} - 1\\
           &= \left( \frac{x'}{x} \right)^2 - 1
\end{align*}
Let $q_\sigma$ the maximum amount of tokens $x$ that can be traded with spread $\sigma$. Injecting $x' = x + q_\sigma$ in the previous equation yields:
\[ 1 + \sigma = \left( 1 + \frac{q_\sigma}{x} \right)^2 \]
and 
\[ q_\sigma = x\left( \sqrt{1 + \sigma} - 1 \right). \]
The rest of the proof is similar to the previous one.
\end{proof}

\begin{corollary}[Realized exchange rate for a swap]
Suppose that a user wants to trade $n$ tokens $\mathbf{Y}$ for some tokens $\mathbf{X}$, the the actual exchange rate $r$ that the user will get is 
\[
    r = \frac{x}{y + n}.
\]
\end{corollary}

\begin{proof}
When there is no maximum spread, the previous equations becomes
\[
    m = \frac{xn}{y + n} .
\]
We have by definition $r = \frac{m}{n}$. Therefore :
\[
    r = \frac{x}{y + n}.
\]
\end{proof}

One could ask whether splitting a single transaction into $k$ smaller transactions would yield more tokens in return. The following theorem states that it is not the case.
\begin{theorem}
Let $k$ be a positive integer and $m_i$ the number of tokens received at the $i^{th}$ transaction $(1 \leq i \leq k)$, then :
\[
    m = \sum_{i=1}^k m_i .
\]
\end{theorem}
\begin{proof}

Let $k$ be a positive integer and $x_i$ (resp. $y_i$) be the number of tokens $\mathbf{X}$ (resp. $\mathbf{Y}$) in the pool after the $i^{th}$ transaction ($1 \leq i \leq k$) with $x$ (resp. $y$) the inital amounts of tokens in the pool. Let $n_i$ be the number of tokens $\mathbf{Y}$ traded at the $i^{th}$ transaction. We have the following relation:
\[
    y_{i} = y_{i-1} + n_i \text{ and } y_0 = y
\]
In this situation, we allow for any spread ie. $\sigma \rightarrow  1$. We can compute the sum : 
\begin{align*}
    \sum_{i=1}^k m_i &= \sum_{i=1}^k \frac{x_{i-1}n_i}{y_{i-1} + n_i} \cmt{using the swap equation with $n_i$ and $\sigma \rightarrow  1$} \\
                       &= \sum_{i=1}^k \frac{xy}{y_{i-1}}\frac{n_i}{y_{i-1} + n_i} \cmt{using $x_{i-1}y_{i-1}     = x_{i-2}y_{i-2} = \dots = xy$} \\
                       &= xy\sum_{i=1}^k \left(\frac{1}{y_{i-1}} - \frac{1}{y_{i-1} + n_i} \right) \\
                       &= xy\sum_{i=1}^k \left( \frac{1}{y_{i-1}} - \frac{1}{y_{i}} \right) \\
                       &= xy \left( \frac{1}{y} - \frac{1}{y_k}\right) \cmt{by telescoping the sum} \\
                       &= x - \frac{xy}{y_k} \\
                       &= x - x_k \cmt{using $xy = x_ky_k$} \\
                       &= m.
\end{align*}
\end{proof}

\subsection{Impermanent Loss} 
Something that is often overlooked is the \textit{impermanent loss} that the providers of the pool can face. It represents the risk that a provider takes by providing liquidity. Therefore the fee $f$ of the pool not only rewards the providers for the service they are offering but also the risk they are willing to take. But what is impermanent loss? Impermanent loss is the relative difference between the value of a portfolio invested in a pool and the value it would have if it had not been invested. It is called \textit{impermanent} because this loss depends on the current rate of exchange of both tokens $\mathbf{X}$ and $\mathbf{Y}$. Hence, the loss can fluctuate and is only realised when providers remove liquidity from the pool. In this section, $x$ and $y$ will represent the initial number of tokens held by the user and $x'$ and $y'$ will represent these same numbers after being invested in the pool and submited to some fluctuations of the exchange rate $r$.

\begin{definition}[Impermanent loss]
Let $V_p$ be the value of the portfolio if invested in the pool and $V_n$ the value of the portfolio if not invested in the pool. Then, the impermanent loss is defined by :
\[
    \Lambda = \frac{V_p - V_n}{V_n}
\]
\end{definition}

\begin{theorem}
Let $\delta_x$ and $\delta_y$ be the relative price changes of $\mathbf{X}$ and $\mathbf{Y}$, then the impermanent loss is : 
\[
    \Lambda = 2\frac{\sqrt{\delta_y\delta_x}}{\delta_x + \delta_y} - 1 ,
\]
with $\delta_x = \frac{p_x'}{p_x}$ and $\delta_y = \frac{p_y'}{p_y}$ .
\end{theorem}
\begin{proof}
In this situation, $V_p$ and $V_n$ are defined by the following equations : 
\begin{align*}
    V_n  &= xp_x' + yp_y' \\
    V_p &= x'p_x' + y'p_y' .
\end{align*}
From our pool model, we have both : 
\begin{equation}
  L^2 =
    \begin{cases}
      \frac{1}{r}x^2  \\
      \frac{1}{r'}(x')^2 
    \end{cases}       
\end{equation}
Let $r'$ be the final exchange rate of the pool. We have :
\[
    x' = \sqrt{\frac{\delta_y}{\delta_x}}x \cmt{because $\frac{r'}{r} = \frac{\delta_y}{\delta_x}$.}
\]
and 
\begin{align*}
    V_p &= x'p_x' + y'p_y' \\
       &= 2x'p_x' \cmt{from the fixed product assumption} \\
       &= 2p_x'\sqrt{\frac{\delta_y}{\delta_x}}x  \\
       &= 2p_x x\sqrt{\delta_y\delta_x} \cmt{using $p_x' = \delta_x p_x$.}
\end{align*}
 We then find an expression of $V_n$ only involving $x$: 
\begin{align*}
    V_n &= xp_x' + yp_y' \\
       &= xp_x' + \frac{xp_x}{p_y}p_y' \\
       &= xp_x' + \delta_yxp_x\\
       &= xp_x\left( \delta_x + \delta_y \right) .
\end{align*}
Computing the difference gives: 
\begin{align*}
V_p - V_n  &= 2p_xx\sqrt{\delta_y\delta_x} - xp_x\left( \delta_x + \delta_y \right)\\
          &= p_xx\left( 2\sqrt{\delta_y\delta_x} -\left( \delta_x + \delta_y \right) \right) .
\end{align*}
The latter can finally be used to compute the impermanent loss : 
\begin{align*}
    \Lambda &= \frac{V_p - V_n}{V_n} \\
          &= \frac{p_xx\left( 2\sqrt{\delta_y\delta_x} -\left( \delta_x + \delta_y \right) \right)}{xp_x\left( \delta_x + \delta_y \right)} \\
          &= 2\frac{\sqrt{\delta_y\delta_x}}{\delta_x + \delta_y} - 1
\end{align*}
\end{proof}

In Figure (\ref{fig:onevaril}) we plot $\Lambda$ when only one of the tokens' price changes. To have a better picture of what is happening we also plot in Figure (\ref{fig:onevarpc})  the relative portfolio comparison. \\ 
These plots are misleading because they give the impression that one is always better off not providing liquidity. This is because the accrued fees are not taken into account in the impermanent loss. 

\subsection{Relative portfolio evolution with time}

As seen in the previous section, impermanent loss does not take into account accrued fees and can therefore be misleading. Before providing liquidity, it is important to understand the winning scenarii and the risks that one is willing to take. In this section, we explore the evolution of a porfolio when providing liquidity or not and according to the way the fees are dealt with. There are 2 ways of dealing with trading fees : 
\begin{enumerate}
    \item Automatically reinject the fees as liquidity (e.g Uniswap v2);
    \item Collect the fees and store them separately (e.g Beaker);
\end{enumerate}
In the rest of the section we respectively reference them as fee model (1) and (2). \\
As swaps occure unpredictably, we consider the evolution according to a given liquidity growth $\alpha$ (\% per year) and we assume, for simplicity, that the growth is linear (ie. that the swaps happen as the same frequence with the same value). 
\begin{theorem}[Relative evolution for fee model (1)]
Let $\delta_x$ and $\delta_y$ be the relative price change of $\mathbf{X}$ and $\mathbf{Y}$ and $t$ the time passed in years. 
Then, in the fee model (1), the relative evolution $\Delta_1(t)=\frac{V_p(t)}{V(t)}$ of the portfolio of a liquidity provider is:
\[
    \Delta_1(t) = \sqrt{\delta_y\delta_x}\left(1 + \alpha t\right) .
\]
\end{theorem}
\begin{proof}
We first compute the number $x(t)$ of tokens $\mathbf{X}$ that should be in the pool if the price does not change.
When  compounding with a liquidity evolution $\alpha$, the liquidity of the pool will grow by $\alpha$ per year. Therefore, we have:  
\[
    L(t) = L(0)\left(1 + \alpha t\right).
\]
When the rate of the pool does not change, both tokens grow equally, therefore:
\begin{align*}
    x(t) = x(0)\left(1 + \alpha t\right) .
\end{align*}

Because, $L$ does not depend on $r$, we can reuse $x' = \sqrt{\frac{\delta_y}{\delta_x}}x$ and the rest of the proof for the impermanent loss. It gives:
\[
    V_p(t) = 2p_xx\sqrt{\delta_y\delta_x}\left(1 + \alpha t\right) .
\]
Dividing by the initial value of the portfolio $V = 2xp_x$ ends the proof.
\end{proof}

\begin{remark}
One could think that with automatic compounding, the value of the portfolio should grow exponentially. In fact, it is only true if we assume that the the exchange volume of the pool grows at the same rate as the liquidity of the pool. However, there is no reason that this should be true. Some exchanges implicitly make this false assumption when they compute the expected APY of the pool.
\end{remark}

\begin{theorem}[Relative evolution for fee model (2)]
Let $\delta_x$ and $\delta_y$ be the relative price change of $\mathbf{X}$ and $\mathbf{Y}$ and $t$ the time passed in years. 
Then, in the fee model (2), the relative evolution $\Delta_2(t)=\frac{V'(t)}{V(t)}$ of the portfolio of a liquidity provider is:
\[
    \Delta_2(t) = \sqrt{\delta_y\delta_x}  + \alpha t\frac{\delta_x+\delta_y}{2} .
\]
\end{theorem}
\begin{proof}
In the fee model (2) the accrued fees are stored separately and not compounded. 
With our assumption on fees accumulation, the accumulated amount of tokens $x_{out}(t)$ in the separate storage is such that : 
\[
    x_{out}(t) = x(0)\alpha t .
\]
Therefore the final value of the portfolio is:
\begin{align*}
    V'(t) &= V_p(t) + V_{out}(t) \\
      &= 2p_xx\sqrt{\delta_y\delta_x} + p_x'x_{out}(t) + p_y'y_{out}(t) \cmt{re-using result from IL}.
\end{align*}
We therefore get :
\begin{align*}
    \Delta_2(t) &= \frac{V'(t)}{V(t)} \\
                       &= \sqrt{\delta_y\delta_x} + \frac{p_x'x_{out}(t) + p_y'y_{out}(t)}{2p_xx} \\
                       &= \sqrt{\delta_y\delta_x} + \left( \frac{\delta_x}{2} + \frac{p_y'y}{2p_yy} \right)\alpha t \cmt{using $p_xx =p_yy $} \\
                       &= \sqrt{\delta_y\delta_x}  +\alpha t\frac{\delta_x+\delta_y}{2} .
\end{align*}
\end{proof}

Figure (\ref{fig:feemodelcomp}) compares the two fee model for a liquidity growth $\alpha = 20\%$. This figure show that fee model (2) is both safer and more profitable than fee model (1). Intuitively, rewards, when compounded, are also prone to impermanent loss.  However, there is a catch: what happens in fee model (2) when some people regularly compound and others don't?

\subsection{Return on Investment}
At the end of the previous section, we asked an important question. To answer it, we are first going to compute the return on investment for users compounding and users not compounding their fees. 
We start by computing the liquidity evolution for people compounding their rewards. 

\begin{theorem}[]
Let $L_c$ (resp. $L_{nc}$) be the liquidity of the pool owned by people compounding (resp. not compounding) their rewards and $L_0$ be the initial liquidity of the pool. Then, $L_c$ satisfies the differential equation:
\[
    \dv{L_c}{t} = \alpha L_0\frac{L_c}{L_c + L_{nc}}.
\]
\end{theorem}

\begin{proof}
By definition $L_{nc}$ is a constant.
The total amount of fees generated at time $t$ is:
\[
    f(t) = \alpha L_0 t .
\]
Let $\mathcal{P}(t)=\frac{L_c(t)}{L_c(t) + L_{nc}}$ the proportion of the pool owned by people compounding, then the growth of $L_c$ is equal to the product of the  growth in fees and $\mathcal{P}(t)$:
\begin{align*}
    \dv{L_c}{t}(t) &= \dv{f}{t}\mathcal{P}(t) \\
                &= \alpha L_0 \frac{L_c(t)}{L_c(t) + L_{nc}}
\end{align*}
\end{proof}

\begin{definition}[Return on investment for compounding users]
For compounding users, we call return on investment the ratio
\[
    \rho_c(t) = \frac{L_c(t)}{L_c(0)} .
\]
\end{definition}

\begin{definition}[Return on investment for not compounding users]
For not compounding users, we call return on investment the ratio
\[
    \rho_{nc}(t) = 1 + \frac{f}{L_{nc}},
\]
where $f$ is the total amount of accrued fees for these users.

\end{definition}

Because the solution of the previous differential equation has no simple expression, an equation solver was used. A worst case scenario (99\% users compounding) is plotted in  Figure (\ref{fig:roicomp}). The plot shows that users compounding have a $20.02\%$ ROI and the others a $18.20\%$ ROI instead of the initial $20\%$ expected ROI. 

\begin{remark}
Because users only lose $1.8 \%$ ROI by not recompounding it is therefore likely that they would not care. 
\end{remark}

From this, it is now possible to compare the portfolio evolution of users that recompound their fees and those who don't. This comparison is made in Figure (\ref{fig:correctfeemodelcomp}). This new comparison shows that users that recompound have a slight hedge when the relative price does not change too much. Still, as pointed out in previous remark, regularly compounding rewards for a little premium does not seem worth it. Therefore, users investing in Beaker pools should expect their portfolio to change as described in (\ref{fig:feemodelcomp}).

\newpage
\section{Conclusion}
To sum up, our analysis has shown that investing in a liquidity pool using a constant-product market making function boils down to making a bet on the price evolution of both tokens. We have also proven that the protocol is safer and more profitable when no one recompounds fees.

\newpage
\section{Figures}
\begin{figure}[H] 
    \centering
    \begin{tikzpicture}[scale = 1.3]
        \begin{axis}[
            axis lines = left,
            xlabel = Relative price change in \%,
            ylabel = Impermanent loss in \%,
            xtick pos=upper,
            xticklabel pos=upper
        ]
            \addplot [
                domain=0:300, 
                samples= 100, 
                color=red,
            ]
            {200*sqrt((x/100))/(1 + x/100) - 100};
            \end{axis}
        \end{tikzpicture}
    \caption{Impermanent loss when only the price of one coin changes}
    \label{fig:onevaril}
\end{figure}
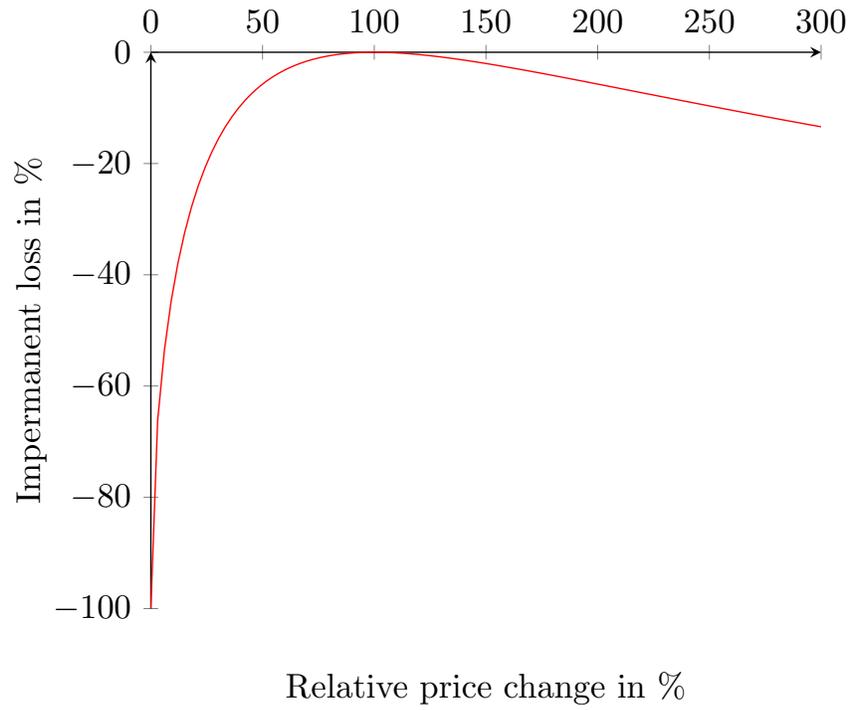

\begin{figure}[H]
    \centering
    \begin{tikzpicture}[scale = 1.3]
        \begin{axis}[
            axis lines = left,
            xlabel = Relative price evolution in \%,
            ylabel =  Relative portfolio evolution in \%,
            legend pos=outer north east,
            legend cell align=left,
        ]
            \addplot [
                domain=-100:200, 
                samples=100, 
                color=red,
            ]
            {100*sqrt((x+100)/100)};
            \addlegendentry{Providing liquidity}
            \addplot [
                domain=-100:200, 
                samples= 100, 
                color=blue,
            ]
            {50*(1+(x+100)/100)};
            \addlegendentry{Not investing}
            \end{axis}
        \end{tikzpicture}
    \caption{Relative portfolio evolution when the relative price of one coin changes}
    \label{fig:onevarpc}
\end{figure}
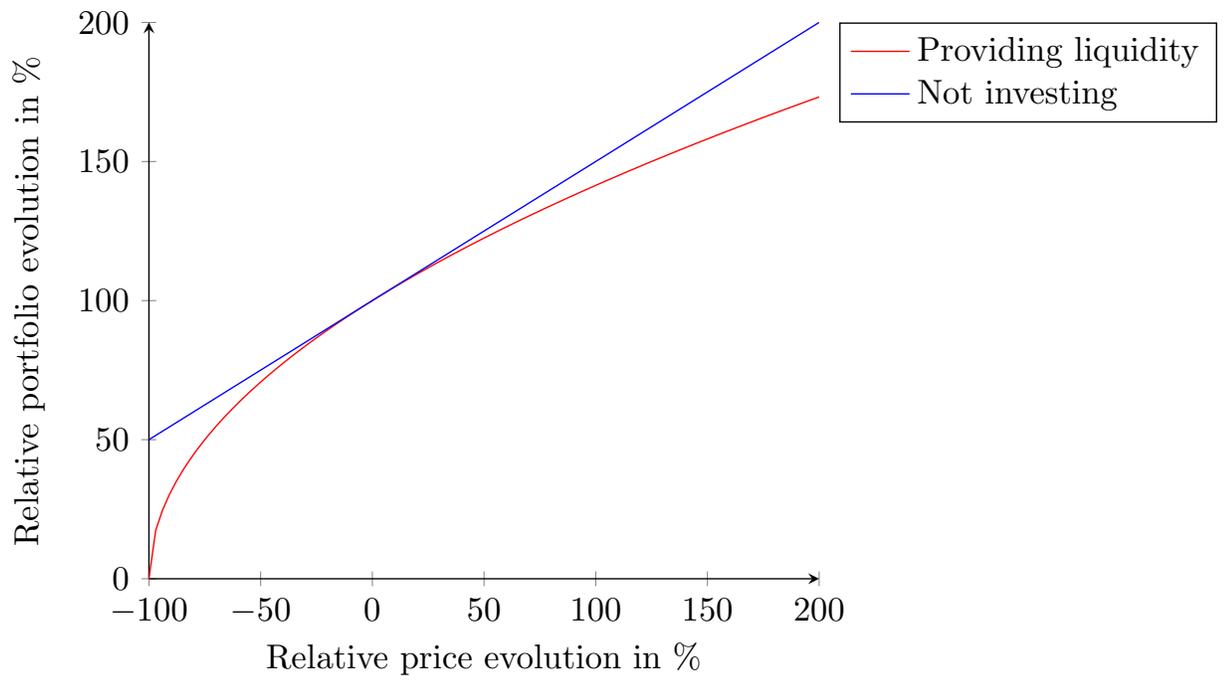

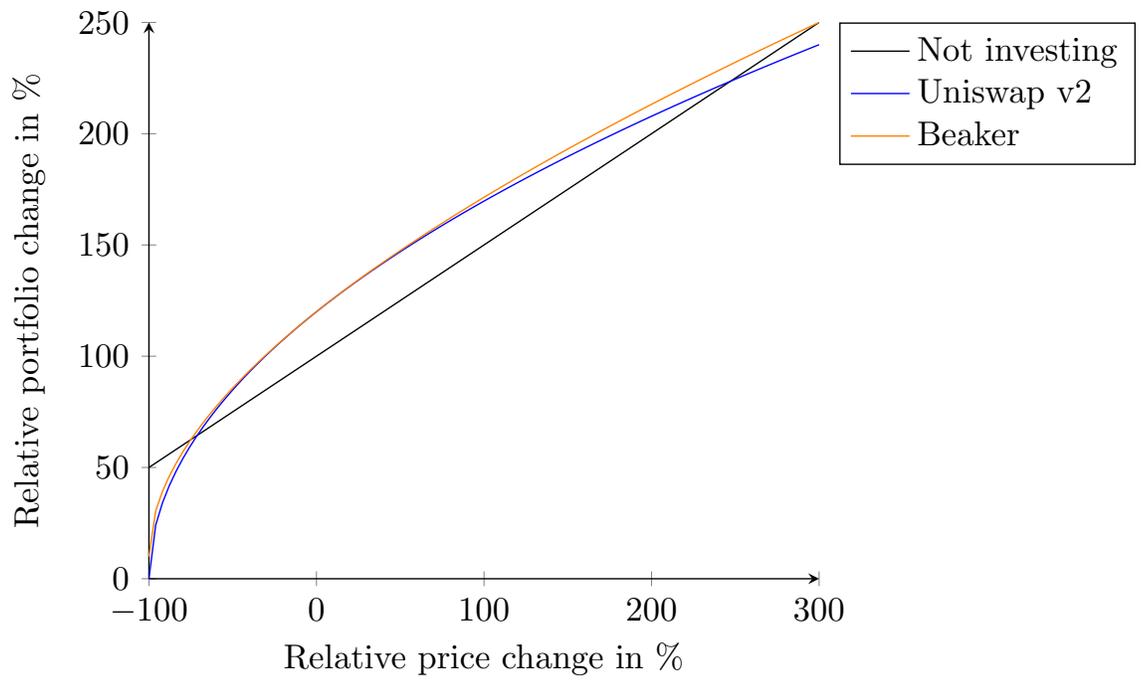
\begin{figure}[H]
    \centering
    \begin{tikzpicture}[scale = 1.3]
        \begin{axis}[
            axis lines = left,
            xlabel = Relative price change in \%,
            ylabel =  Relative portfolio change in \%,
            legend pos=outer north east,
            legend cell align=left,
        ]
            \addplot [
                domain=-100:300, 
                samples=100, 
                color=black,
            ]
            {50*(1+(x+100)/100)};
            \addlegendentry{Not investing}

            \addplot [
                domain=-100:300, 
                samples=100, 
                color=blue,
            ]
            {120*sqrt((x+100)/100)};
            \addlegendentry{Uniswap v2}

            \addplot [
                domain=-100:300, 
                samples=100, 
                color=orange
            ]
            {100*sqrt((x+100)/100) + 20*((x+100)/100+1)/2};
            \addlegendentry{Beaker}
            
            \end{axis}
        \end{tikzpicture}
    \caption{Relative portfolio change depending on the pool model after a year when the price of one coin changes and with 20 \% APR}
    \label{fig:feemodelcomp}
\end{figure}

\begin{figure}[H]
    \centering
    \begin{tikzpicture}[scale = 1.3]
         \begin{axis}[
            axis lines = left,
            xlabel = Time in years,
            ylabel =  ROI,
            legend pos=outer north east,
            legend cell align=left,
        ]
        
            \addplot[
                color=blue,
            ]
            table {figures/data/comp_roi.txt};
            \addlegendentry{Compounding}
            
            \addplot[
                color=red,
            ]
            table {figures/data/not_comp_roi.txt};
            \addlegendentry{Not compounding}
        \end{axis}
    \end{tikzpicture}
    \caption{ROI comparison when 99\% of users recompound their rewards and for an initial 20\% APR}
    \label{fig:roicomp}
\end{figure}
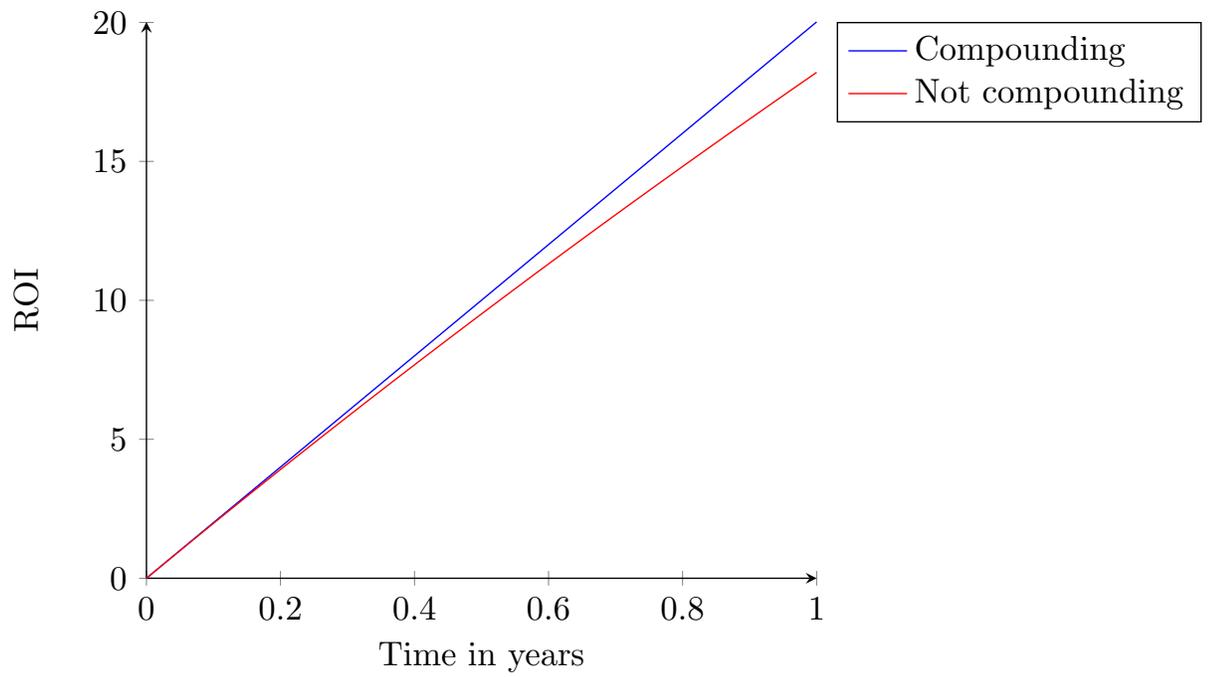

\begin{figure}[H]
    \centering
    \begin{tikzpicture}[scale = 1.3]
        \begin{axis}[
            axis lines = left,
            xlabel = Relative price change in \%,
            ylabel =  Relative portfolio change in \%,
            legend pos=outer north east,
            legend cell align=left,
        ]
            \addplot [
                domain=-100:150, 
                samples=100, 
                color=black,
            ]
            {50*(1+(x+100)/100)};
            \addlegendentry{Not investing}
            
            \addplot [
                domain=-100:150, 
                samples=100, 
                color=blue,
            ]
            {120.02*sqrt((x+100)/100)};
            \addlegendentry{Compounding}

            \addplot [
                domain=-100:150, 
                samples=100, 
                color=pink
            ]
            {100*sqrt((x+100)/100) + 18.2*((x+100)/100+1)/2};
            \addlegendentry{Not compounding}
            
            \end{axis}
        \end{tikzpicture}
    \caption{Corrected relative portfolio change depending on compounding fees after a year when the price of one coin changes and with 20 \% initial APR}
    \label{fig:correctfeemodelcomp}
\end{figure}
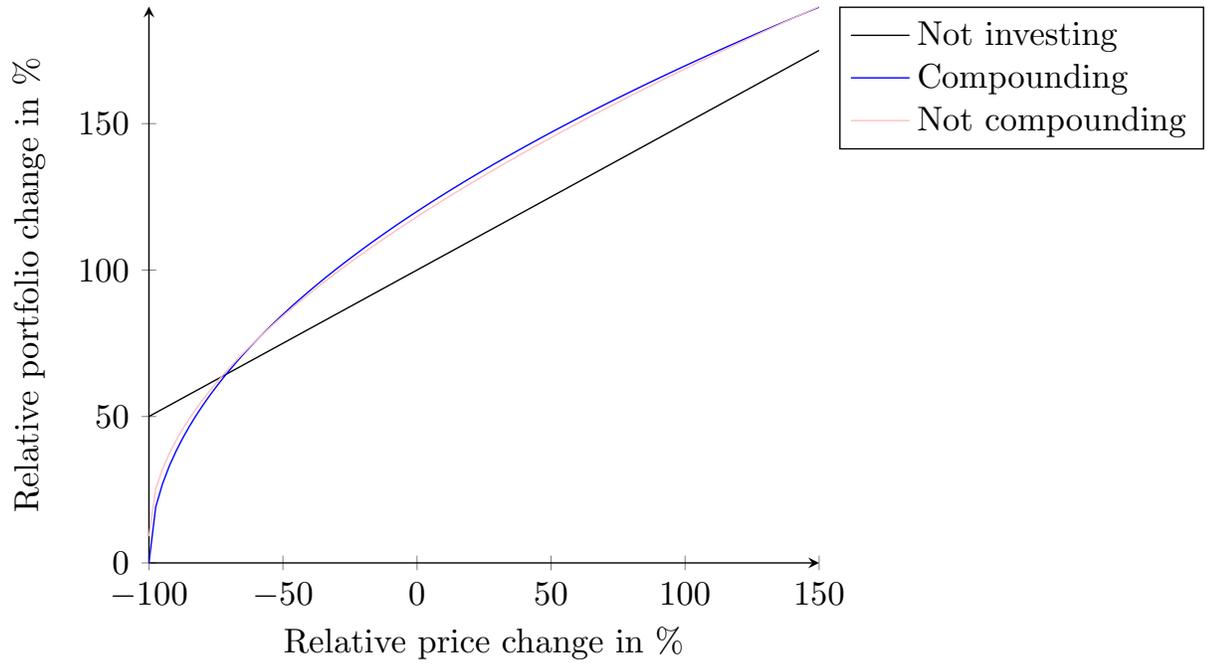

\bibliography{ref}

\begin{thebibliography}{1}

\bibitem{defilama}
Defi tvl tracker.
\newblock \url{https://defillama.com/}.

\bibitem{UniswapV1}
Hayden Adams.
\newblock Uniswap v1.
\newblock \url{https://docs.uniswap.org/protocol/V1/introduction}.

\bibitem{ethereum}
Vitalik Buterin.
\newblock Ethereum whitepaper.
\newblock \url{https://ethereum.org/en/whitepaper/}.

\bibitem{makerdao}
Rune Christensen.
\newblock Makerdao whitepaper.
\newblock \url{https://makerdao.com/en/whitepaper/#abstract}.

\bibitem{Curve}
Michael Egorov.
\newblock Curve stableswap pool whitepaper.
\newblock \url{https://curve.fi/files/stableswap-paper.pdf}.

\bibitem{UniswapV2}
Dan~Robinson Hayden~Adams, Noah~Zinsmeister.
\newblock Uniswapv2 core whitepaper.
\newblock \url{https://uniswap.org/whitepaper.pdf}.

\bibitem{UniswapV3}
Moody~Salem Hayden~Adams, Noah~Zinsmeister.
\newblock Uniswapv3 core whitepaper.
\newblock \url{https://uniswap.org/whitepaper-v3.pdf}.

\bibitem{bitcoin}
Satoshi Nakamoto.
\newblock Bitcoin whitepaper.
\newblock \url{https://bitcoin.org/bitcoin.pdf}.

\end{thebibliography}
\bibliographystyle{plain}

\end{document}